\documentclass{article}
\usepackage{spconf,amsmath,graphicx,hyperref}

\usepackage{color}
\usepackage{graphicx,subfigure,amsmath,amssymb,amsfonts,bm,epsfig,epsf,url,dsfont}
\usepackage{amsthm,mathrsfs}

\newtheorem{thm}{Theorem}[section]
\newtheorem{lem}[thm]{Lemma}

\newtheorem{proposition}[thm]{Proposition}

\usepackage{tikz}
\usepackage{bbm}
\usepackage[small,bf]{caption}
\usepackage{fancybox}
\usepackage{hyperref}
\usepackage{algorithm}
\usepackage{algorithmic}
\usepackage{verbatim}
\usepackage[titletoc,toc]{appendix}
\usepackage{relsize}

\usepackage{mathtools}

\usepackage{scalerel,stackengine}
\stackMath
\newcommand\reallywidehat[1]{%
\savestack{\tmpbox}{\stretchto{%
  \scaleto{%
    \scalerel*[\widthof{\ensuremath{#1}}]{\kern-.6pt\bigwedge\kern-.6pt}%
    {\rule[-\textheight/2]{1ex}{\textheight}}
  }{\textheight}%
}{0.5ex}}%
\stackon[1pt]{#1}{\tmpbox}%
}

\newcommand*{\rom}[1]{\expandafter\@slowromancap\romannumeral #1@}

\DeclareMathOperator*{\argmax}{arg\,max}

\usepackage{xspace}

\numberwithin{equation}{section}

\allowdisplaybreaks

\title{A Geometric Analysis of Initialization Bias in Spherical $K$-means in the Weak Signal Regime}
\name{Amnon Balanov and Tamir Bendory \thanks{This research is  supported in part by NSF-BSF under Grant 2024791, and in part by ISF under Grant 1046/26.}}

\address{Department of Electrical and Computer Engineering\\
Technion---Israel Institute of Technology, Haifa 3200003, Israel}

\begin{document}
%





\maketitle
\begin{abstract}
We study initialization bias in spherical $K$-means for weakly informative directional mixtures. We model the observations by a $K$-component von Mises-Fisher mixture with a small concentration parameter $\kappa$, corresponding to a high-dispersion regime in which the data provide limited information about the underlying directions. Our analysis begins with the limiting case $\kappa=0$ (corresponding to a uniform distribution over the sphere), where one population spherical $K$-means update is governed entirely by the Voronoi tessellation induced by the initialized templates. For uniformly random initializations in fixed dimension $d$, the updated templates become asymptotically aligned with their initial values as $K\to\infty$: the average squared geodesic error scales as $O(K^{-2/(d-1)})$, while the worst-case error is $O((\log K/K)^{2/(d-1)})$. We then show that, in the weak-signal regime of small positive $\kappa$, the population update remains an $O(\kappa)$ perturbation of this limiting map. Thus, in the weak-signal regime, spherical $K$-means can preserve initialization-induced structure despite the presence of a genuine but highly dispersed directional signal.
\end{abstract}


\section{Introduction}

Confirmation bias is the tendency to interpret information in a way that supports prior beliefs while discounting incompatible evidence, a phenomenon well documented in psychology, medicine, and forensic science~\cite{klayman1995varieties,kassin2013forensic}. A related effect can arise in computational and statistical inference. Many latent-variable problems are nonconvex, so they are often solved by iterative algorithms whose output depends on initialization. In practice, these algorithms are frequently initialized with hypotheses intended to improve convergence. However, the initialization can act as a form of prior belief: it determines the first assignments or alignments, and these early decisions can influence subsequent updates. If the data are weakly informative or contain no actual signal, the algorithm may fail to correct the initialized hypotheses, and the final output may remain aligned with them.

Such effects have been discussed in structural biology and related reconstruction pipelines~\cite{henderson2013avoiding,shatsky2009method,sigworth1998maximum,balanov2026einstein,balanov2026confirmation}, and were recently formalized for Gaussian mixture models under pure-noise observations~\cite{balanov2025confirmation}. In that setting, even isotropic Gaussian noise can lead the standard $K$-means and expectation-maximization (EM) algorithms to produce estimates positively correlated with their initialization.

\noindent\textbf{Spherical K-means.} Here, we study the corresponding mechanism for normalized directional data, where both observations and candidate hypotheses are represented as directions (that is, points on the sphere). This setting arises in many estimation and clustering problems, and is particularly relevant in structural biology. In cryogenic electron microscopy (cryo-EM) and related reconstruction problems, iterative refinement procedures repeatedly compare particle images with class averages or projections of the current estimate using correlation-based scores, and then update the estimate from the resulting assignments and alignments~\cite{cheng2015primer,sigworth2010maximum, bendory2020single}. After normalization, these comparisons have a natural spherical geometry: each observation defines a direction, and the algorithm assigns it to the most aligned candidate. Spherical $K$-means provides a minimal model of this assignment-update mechanism~\cite{dhillon2001concept, banerjee2005clustering, hornik2012spherical}. It isolates the geometric question at the center of this work: when the data are weakly informative, or contain no directional signal, does the algorithm correct the initialization, or does the spherical update preserve it? 

\begin{figure*}[t!]
    \centering
    \includegraphics[width=0.67 \linewidth]{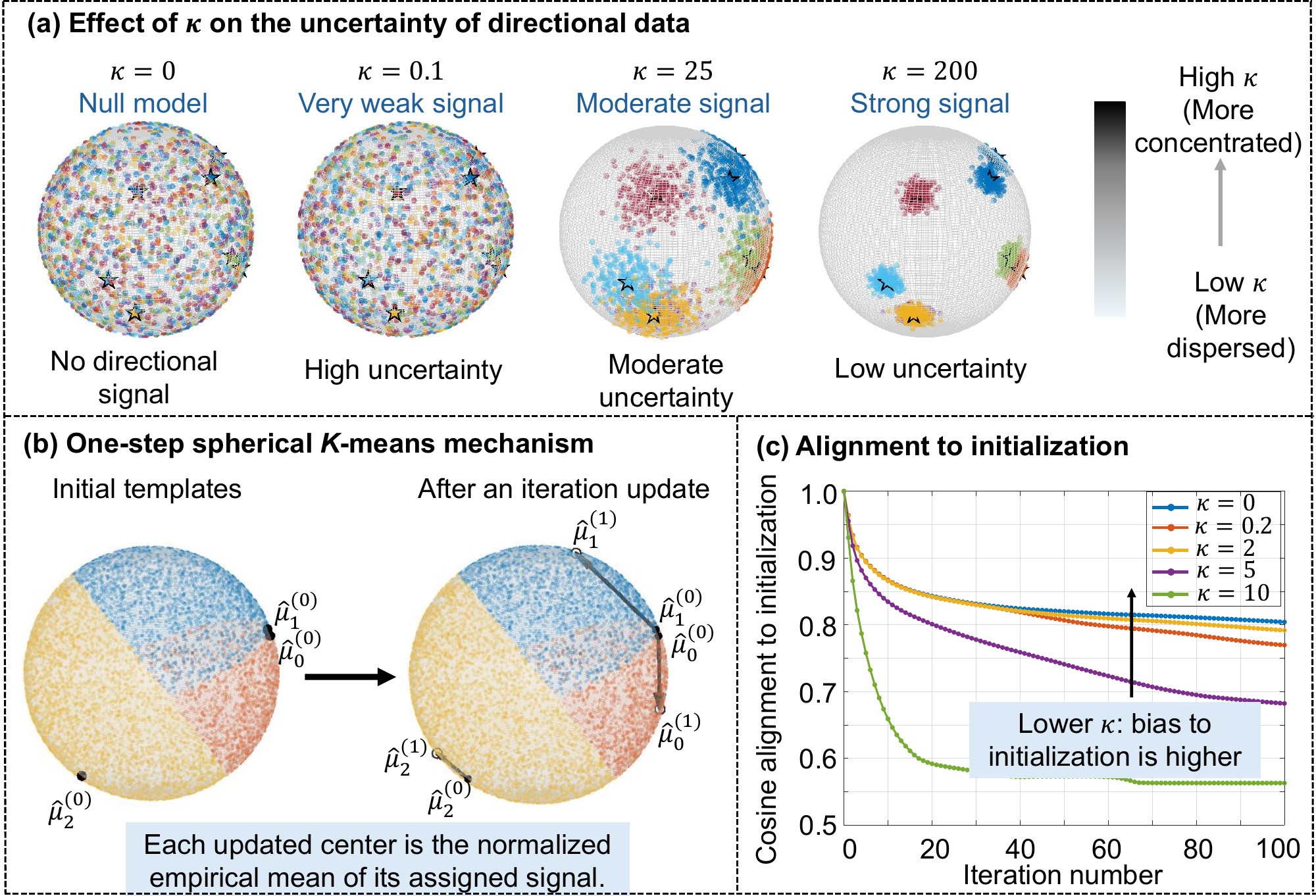}
    \caption{\textbf{Weak directional signal and initialization bias in spherical $K$-means.}
    (a) Illustration of a von Mises-Fisher mixture on the sphere. The concentration parameter $\kappa$ controls the amount of directional information around the mixture centers: $\kappa=0$ gives the uniform model, while small positive $\kappa$ gives weakly informative observations.
    (b) Given initial templates $\{{\mu}^{(0)}_\ell\}_{\ell=0}^{K-1}$, spherical $K$-means partitions the sphere into Voronoi cells and updates each template to the normalized mean of its assigned observations. As $K$ increases, the Voronoi cells become smaller, and their normalized centroids become increasingly aligned with the templates that generated them.
    (c) Numerical illustration, for $K=20$, of the mean cosine alignment with the initialization as a function of the iteration number. The alignment remains high for small $\kappa$, suggesting that spherical $K$-means can retain initialization-induced structure in weakly informative mixtures.}
    \label{fig:1} 
\end{figure*}

\noindent\textbf{The experiment.}
We analyze this question through the $K$-component von Mises-Fisher mixture model~\cite{fisher1953dispersion,mardia2009directional,banerjee2005clustering}, a standard probabilistic model for directional data. Each component is centered at a direction, while the concentration parameter $\kappa$ controls how strongly the observations are concentrated around the component directions. Our primary interest is the high-dispersion regime of small positive $\kappa$, in which the observations carry only weak directional information; see Figure~\ref{fig:1}(a).

To isolate the effect of initialization and obtain an analytically tractable reference point, we begin with the limiting case $\kappa=0$, in which the vMF mixture reduces to the uniform distribution on $\mathbb{S}^{d-1}$. In this limiting experiment, spherical $K$-means is initialized with $K$ template directions, interpreted as hypothesized component centers, and each observation is assigned to its most aligned template. Although the observations are rotationally invariant, the initialization breaks this symmetry by inducing a spherical Voronoi tessellation. In the population limit, each template is updated toward the normalized centroid of the Voronoi cell that it generates. Consequently, the output can remain strongly aligned with the initialized hypotheses rather than being determined by directional information in the observations. This provides a geometric mechanism for model bias and a baseline against which the behavior under small positive $\kappa$ can be analyzed.

\noindent\textbf{Main results and organization.}
Our main analysis focuses on a single population spherical $K$-means update, which is already sufficient to expose the fundamental initialization-induced mechanism under the uniform model. Section~\ref{subsec:connection_gaussian_confirmation_bias} establishes the structural properties of this one-step update and relates the spherical model to the previously studied Gaussian confirmation-bias model~\cite{balanov2025confirmation}. This Gaussian-spherical relation allows us to transfer several structural results to spherical $K$-means. In particular, we characterize the population update as the normalized Voronoi-cell centroid, and show asymptotic alignment when both the dimension and the number of templates grow under a near-orthogonality condition. These results provide the geometric and probabilistic foundations for our fixed-dimensional analysis.

Our main contribution, developed in Section~\ref{sec:fixed_dimension}, concerns the complementary regime in which the dimension $d$ is fixed and the number of postulated mixture components $K$ grows. For uniformly random initial templates, the induced Voronoi tessellation becomes progressively finer, forcing the normalized cell centroids toward the templates that generated them. We quantify this effect through two complementary rates. The average squared geodesic alignment error scales as $O(K^{-2/(d-1)})$; see Proposition~\ref{prop:typical_cell_rate}. The worst-case squared error is controlled by the spherical covering radius and satisfies $O((\log K/K)^{2/(d-1)})$; see Theorem~\ref{thm:hardAssignmentAsymptoticLandFixedD}. Thus, increasing the number of postulated components can make spherical $K$-means preserve the initialization more strongly.

Section~\ref{sec:extensions} develops several consequences and extensions of this geometric analysis. In particular, we show that for small positive $\kappa$, a vMF mixture is a perturbation of the uniform distribution on the sphere, and the corresponding population update remains close to the limiting Voronoi-centroid map. Hence, when the directional signal is weak, the geometry induced by the initialization remains a leading contribution to the update. In addition, for $d=2$, we obtain an explicit characterization of the multi-iteration population dynamics on the circle. The cyclic gaps between neighboring templates evolve according to a discrete heat equation, so the templates converge to a uniformly spaced configuration. The corresponding spectral gap is of order $K^{-2}$, consistent with the $d=2$ specialization of the one-step alignment rate, and yields a relaxation time of order $K^2$.

\section{Problem formulation}
\label{sec:problemFormulation}

This section introduces the directional model and the spherical $K$-means procedure studied in this work.

\subsection{The von Mises-Fisher mixture model}
\label{subsec:vMFandour}

Let $\mathbb{S}^{d-1}=\{v\in\mathbb{R}^d:\|v\|_2=1\}$ denote the unit sphere, and let $d\sigma(v)$ denote its normalized surface measure, so that $\sigma(\mathbb{S}^{d-1})=1$. For a mean direction $\mu\in\mathbb{S}^{d-1}$ and concentration parameter $\kappa\geq0$, the vMF density with respect to $d\sigma$ is
\begin{align}
    f_{\mathrm{vMF}}(v;\mu,\kappa) = c_d(\kappa)\exp\{\kappa \mu^\top v\},
    \label{eqn:vmf_density}
\end{align}
for $v\in\mathbb{S}^{d-1}$, where $c_d(\kappa)$ is the normalizing constant. In particular, $c_d(0)=1$, so $\kappa=0$ gives the uniform distribution on $\mathbb{S}^{d-1}$, whereas larger values of $\kappa$ correspond to stronger concentration around $\mu$.

A $K$-component vMF mixture is specified by weights $\boldsymbol{w}=(w_0,\ldots,w_{K-1})$, mean directions $\boldsymbol{\mu}=(\mu_0,\ldots,\mu_{K-1})$, and concentrations $\boldsymbol{\kappa}=(\kappa_0,\ldots,\kappa_{K-1})$. Here, $w_\ell\geq0$, $\sum_{\ell=0}^{K-1}w_\ell=1$, $\mu_\ell\in \mathbb{S}^{d-1}$, and $\kappa_\ell\geq0$.
Its density is
\begin{align}
    \mathcal{M}^{\mathrm{vMF}}_{K,d}(\boldsymbol{w},\boldsymbol{\mu},\boldsymbol{\kappa})(v) \triangleq \sum_{\ell=0}^{K-1} w_\ell f_{\mathrm{vMF}}(v;\mu_\ell,\kappa_\ell).
    \label{eq:vMF-mixture-density}
\end{align}

\subsection{The uniform limiting experiment}

We consider first the experiment in which the true observations carry no directional cluster structure and are sampled i.i.d.\ from the uniform distribution:
\begin{align}
    v_0,\ldots,v_{n-1} \stackrel{\mathrm{i.i.d.}}{\sim} \mathrm{Unif}(\mathbb{S}^{d-1}).
    \label{eq:true-uniform-model}
\end{align}
The researcher nevertheless fits a postulated $K$-component directional model with equal weights $1/K$ and a shared concentration parameter $\kappa$:
\begin{align}
    v_0,\ldots,v_{n-1}    \stackrel{\mathrm{i.i.d.}}{\sim}    \mathcal{M}^{\mathrm{vMF}}_{K,d}    (\mathbf{1}/K,\boldsymbol{\mu},\kappa\mathbf{1}).
    \label{eq:postulated-vmf-model}
\end{align}
Thus, the fitted model assumes that the observations are concentrated around $K$ component directions, even though the true distribution in~\eqref{eq:true-uniform-model} is uniform. 

The researcher seeks to estimate the component directions $\boldsymbol{\mu}=(\mu_0,\ldots,\mu_{K-1})$ using spherical $K$-means.
To initialize the algorithm, the researcher specifies a collection of hypothesized component directions, ${\boldsymbol{\mu}}^{(0)} = ( {\mu}^{(0)}_0,\ldots, {\mu}^{(0)}_{K-1})$. Our goal is to understand how this initialization influences the resulting estimates of the component directions and, in particular, whether the observations correct the initial hypotheses or whether the estimated directions remain systematically aligned with them.

\subsection{Spherical $K$-means}
\label{subsec:spherical_kmeans}

Spherical $K$-means is the standard hard-assignment algorithm for equal-weight vMF mixture estimation with a common concentration parameter~\cite{hornik2012spherical}. Starting from $K$ initial directions on the sphere, the algorithm alternates between two steps. First, each observation is assigned to the template with which it has the largest inner product. Second, each template is updated to the normalized empirical mean of the observations assigned to it. Since all observations and templates have unit norm, maximizing the inner product is equivalent to nearest-neighbor assignment in angular distance. Thus, spherical $K$-means converts an assignment-induced partition of the sphere into updated template directions.
Algorithm~\ref{alg:generalizedEfNhard} outlines the iterative procedure.

\begin{algorithm}[t!]
\caption{\texttt{Spherical $K$-means}}
\label{alg:generalizedEfNhard}
\textbf{Input:} Initial templates $\{\mu_\ell^{(0)}\}_{\ell=0}^{K-1}\subset \mathbb{S}^{d-1}$, observations $\{v_i\}_{i=0}^{n-1}\subset \mathbb{S}^{d-1}$, and number of iterations $T$.\\
\textbf{Output:} Updated estimates $\{\mu_\ell^{(T)}\}_{\ell=0}^{K-1}$.

\begin{enumerate}
    \item For $t=0,1,\ldots,T-1$:
    \begin{enumerate}
        \item For each $\ell=0,\ldots,K-1$, initialize the assignment set $\mathcal{A}_\ell^{(t)} \leftarrow \emptyset$. 

        \item For each observation $v_i$, compute
        \begin{align}
            {\mathrm{R}}_i^{(t)} \triangleq \argmax_{0\leq \ell\leq K-1} \langle v_i,\mu_\ell^{(t)}\rangle ,
            \label{eq:hard-assignment-rule}
        \end{align}
        and assign $\mathcal{A}_{{\mathrm{R}}_i^{(t)}}^{(t)} \leftarrow \mathcal{A}_{{\mathrm{R}}_i^{(t)}}^{(t)} \cup \{v_i\}$.

        \item For each $\ell=0,\ldots,K-1$, update
        \begin{align}
            \mu_\ell^{(t+1)} = \frac{\sum_{v_i\in\mathcal{A}_\ell^{(t)}} v_i}{\|\sum_{v_i\in\mathcal{A}_\ell^{(t)}} v_i\|_2}.
            \label{eq:hard-assignment-estimator}
        \end{align}
    \end{enumerate}
\end{enumerate}
\end{algorithm}

\subsection{Notation convention}
Our main analysis focuses on a single iteration, $T=1$. The algorithm is initialized at fixed directions $\mu_0^{(0)},\ldots,\mu_{K-1}^{(0)}$ and performs one assignment-update step. For notational brevity, throughout the single-step analysis we write $\mu_\ell \coloneqq \mu_\ell^{(0)}$ for the initialized direction and denote the resulting update by $\widehat{\mu}_\ell \coloneqq \mu_\ell^{(1)}$. We further denote by $\mu_\ell^\star$ the corresponding one-step population spherical $K$-means update, namely, the limit of $\widehat{\mu}_\ell$ as $n \to \infty$.

For fixed templates, let $\mathcal{V}_\ell$ denote the spherical Voronoi cell generated by $\mu_\ell$:
\begin{align}
    \mathcal{V}_\ell = \left\{v\in \mathbb{S}^{d-1}:\langle v,\mu_\ell\rangle \geq \langle v,\mu_k\rangle \text{ for all } k \right\}.
    \label{eqn:VlDef}
\end{align}
Throughout, we write $d_S(x,y) \triangleq \arccos\langle x,y\rangle$ for the geodesic distance on $\mathbb{S}^{d-1}$. 

\section{Relation to the Gaussian model and fundamental properties}
\label{subsec:connection_gaussian_confirmation_bias}

This section establishes the structural foundation for the geometric analysis developed in Section~\ref{sec:fixed_dimension}. It has two complementary goals. First, we show that the uniform spherical experiment is the directional analogue of the pure-noise Gaussian confirmation-bias model studied in~\cite{balanov2025confirmation}.
Second, leveraging this correspondence, we establish several fundamental properties of the spherical $K$-means population update that underpin the fixed-dimensional analysis developed in Section~\ref{sec:fixed_dimension}.

The following proposition is the spherical analogue of~\cite[Lemma~A.1]{balanov2025confirmation}: the empirical $K$-means converges to the conditional centroid of the corresponding Voronoi cell, after which spherical $K$-means normalizes this centroid. The result follows directly by applying the strong law of large numbers
to~\eqref{eq:hard-assignment-estimator}.

\begin{proposition}[Population limit of one spherical $K$-means step]
\label{prop:hardAssignmentPopulationLimit}
Fix $d\geq2$ and $K \geq 2$, and let $\mu_0,\ldots,\mu_{K-1}\in \mathbb{S}^{d-1}$ be the pairwise distinct initial templates. Let $\{\widehat{\mu}_\ell\}_{\ell=0}^{K-1}$ be the output of one iteration of Algorithm~\ref{alg:generalizedEfNhard}, based on observations $v_0,\ldots,v_{n-1}\stackrel{\mathrm{i.i.d.}}{\sim}\mathrm{Unif}(\mathbb{S}^{d-1})$. Then, for every $\ell=0,\ldots,K-1$,
\begin{align}
    \widehat{\mu}_\ell \xrightarrow[n\to\infty]{\mathrm{a.s.}} \mu_\ell^\star \triangleq    \frac{\int_{\mathcal{V}_\ell} v\,d\sigma(v)}    {\left\|\int_{\mathcal{V}_\ell} v\,d\sigma(v)\right\|_2}.
    \label{eq:population-voronoi-centroid}
\end{align}
\end{proposition}

\subsection{Gaussian-spherical direction equivalence}

The uniform spherical experiment is the directional counterpart of the pure-noise Gaussian model studied in~\cite{balanov2025confirmation}. Let $g\sim\mathcal{N}(0,I_d)$ and write $u=g/\|g\|_2$. Then $u\sim\operatorname{Unif}(\mathbb{S}^{d-1})$, and for unit-norm templates $\mu_0,\ldots,\mu_{K-1}$,
\begin{align}
    \argmax_{0\leq k\leq K-1}\langle g,\mu_k\rangle = \argmax_{0\leq k\leq K-1}\langle u,\mu_k\rangle .    \label{eqn:direction_equivalence}
\end{align}
Thus, Gaussian and spherical hard assignment induce the same Voronoi partition.

For each cell $\mathcal{V}_\ell$, define the conditional spherical centroid
\begin{align}
    \notag m_\ell^{\mathrm{S}} & \triangleq \mathbb{E}_{u \sim \mathrm{Unif}(\mathbb{S}^{d-1})}\!\left[u \mid u\in\mathcal{V}_\ell\right] 
    \\ & = \frac{ \mathbb{E}_{u \sim \mathrm{Unif}(\mathbb{S}^{d-1})}\!\left[u\mathbf{1}_{\{u\in\mathcal{V}_\ell\}}\right]}{\mathbb{P}(u\in\mathcal{V}_\ell)}.
\end{align}
Since $u$ is uniformly distributed on the sphere, $m_\ell^{\mathrm{S}} = \big[\int_{\mathcal{V}_\ell} u\,d\sigma(u) \big] / \sigma(\mathcal{V}_\ell) $. Hence, the spherical population update in~\eqref{eq:population-voronoi-centroid} is precisely the normalized conditional centroid, $\mu_\ell^\star = m_\ell^{\mathrm{S}}/\|m_\ell^{\mathrm{S}}\|_2$.
Similarly, define the Gaussian hard-assignment mean
\begin{align}
    \notag m_\ell^{\mathrm{G}} & \triangleq \mathbb{E}_{g\sim \mathcal{N}(0,I_d)} \!\left[g \mid g/\|g\|_2\in\mathcal{V}_\ell\right]
     \\ & = \frac{ \mathbb{E}_{g \sim \mathcal{N}(0,I_d)}\!\left[g\mathbf{1}_{\{g/\|g\|_2\in\mathcal{V}_\ell\}}\right]}{\mathbb{P}(g/\|g\|_2\in \mathcal{V}_\ell)}.     \label{eqn:Gaussian_hardAssign_mean}
\end{align}
The following lemma shows that the Gaussian and spherical conditional means differ only by a positive radial factor.

\begin{lem}[Gaussian-spherical direction equivalence]
\label{cor:spherical_gaussian_same_direction}
For every $\ell$,
\begin{align}
    m_\ell^{\mathrm{G}} = \mathbb{E}\|g\|_2\, m_\ell^{\mathrm{S}}.    \label{eqn:directional_equivalence}
\end{align}
Consequently,
\begin{align}
    \mu_\ell^\star = \frac{m_\ell^{\mathrm{S}}}{\|m_\ell^{\mathrm{S}}\|_2} = \frac{m_\ell^{\mathrm{G}}}{\|m_\ell^{\mathrm{G}}\|_2}.
\end{align}
\end{lem}

\begin{proof}
Since $g=\|g\|_2u$, where $\|g\|_2$ is independent of $u$, and since the assignment event depends only on the direction $u$, we have
\begin{align}
    \mathbb{E}_{g \sim \mathcal{N}(0, I_d)}
    \left[g\mathbf{1}_{\{g/\|g\|_2\in \mathcal{V}_\ell\}} \right]
    &=
    \mathbb{E} \left[ \|g\|_2u\mathbf{1}_{\{u\in \mathcal{V}_\ell\}} \right]  \nonumber\\
    &=
    \mathbb{E}\|g\|_2 \cdot \mathbb{E} \left[ u\mathbf{1}_{\{u\in \mathcal{V}_\ell\}} \right].
\end{align}
Moreover, $\mathbb{P}(g/\|g\|_2\in \mathcal{V}_\ell) = \mathbb{P}(u\in \mathcal{V}_\ell)$. Dividing the two identities, and normalizing both sides gives~\eqref{eqn:directional_equivalence}.
\end{proof}

\subsection{Voronoi consistency of the population update}

The next result shows that the population update remains in the Voronoi cell of its generating template; equivalently, it is at least as aligned with that template as with any other template.

\begin{lem}[Voronoi consistency of the population update]
\label{lem:population_update_in_voronoi_cell}
Let $\mu_0,\ldots,\mu_{K-1}\in \mathbb{S}^{d-1}$ be the pairwise distinct initial templates.
Then, for every $\ell$ and every $k$,
\begin{align}
    \langle\mu_\ell^\star,\mu_\ell\rangle \geq \langle\mu_\ell^\star,\mu_k\rangle.    \label{eq:population_update_voronoi_consistency}
\end{align}
\end{lem}

\begin{proof}
This is the spherical counterpart of the hard-assignment consistency property in~\cite[Theorem~3.1(4)]{balanov2025confirmation}.
Indeed, the unnormalized population update is the conditional mean over the assignment region $\mathcal{V}_\ell$, and every point in this region satisfies $\langle v,\mu_\ell\rangle \geq \langle v,\mu_k\rangle$. Averaging over $\mathcal{V}_\ell$ preserves this inequality, and normalization by the positive scalar $\| \int_{\mathcal{V}_\ell} v\,d\sigma(v)\|_2$ yields~\eqref{eq:population_update_voronoi_consistency}.
\end{proof}

\subsection{Closed form for two templates}
The next proposition gives a closed-form expression for the population update in the special case of two templates. It is the spherical analogue of~\cite[Theorem~III.7]{balanov2025confirmation}.

\begin{proposition}[Closed form for two templates]
\label{prop:two_template_closed_form}
Let $K=2$, and let $\mu_0,\mu_1\in \mathbb{S}^{d-1}$ satisfy $\rho=\langle \mu_0,\mu_1\rangle<1$. Then, the one-step population spherical $K$-means updates under the uniform model are
\begin{align}
    \mu_0^\star = \frac{\mu_0-\mu_1}{\|\mu_0-\mu_1\|_2},
    \qquad
    \mu_1^\star = \frac{\mu_1-\mu_0}{\|\mu_1-\mu_0\|_2}.
\end{align}
In particular, $\langle \mu_0^\star,\mu_0\rangle = \langle \mu_1^\star,\mu_1\rangle = \sqrt{(1-\rho)/2}$.
\end{proposition}

\begin{proof}
The Voronoi cell of $\mu_0$ is
\begin{align}
    \mathcal{V}_0
    &=
    \{v\in \mathbb{S}^{d-1}:
    \langle v,\mu_0\rangle
    \geq
    \langle v,\mu_1\rangle\} \nonumber\\
    &=
    \{v\in \mathbb{S}^{d-1}:
    \langle v,\mu_0-\mu_1\rangle
    \geq 0\}.
\end{align}
Thus, $\mathcal{V}_0$ is the hemisphere with pole $\mu_0-\mu_1$. By symmetry, the centroid of a uniform hemisphere points in the direction of its pole. Therefore
\begin{align}
    \mu_0^\star = \frac{\mu_0-\mu_1}{\|\mu_0-\mu_1\|_2}.
\end{align}
The expression for $\mu_1^\star$ follows similarly.
Finally,
\begin{align}
    \langle \mu_0^\star,\mu_0\rangle
    &=
    \frac{\langle \mu_0-\mu_1,\mu_0\rangle}
    {\|\mu_0-\mu_1\|_2} \nonumber\\
    &=
    \frac{1-\rho}{\sqrt{2(1-\rho)}} = \sqrt{\frac{1-\rho}{2}},
\end{align}
and the same calculation applies to $\mu_1^\star$.
\end{proof}

\subsection{Growing number of templates and dimension}
\label{subsec:growing_K_and_d}

The preceding results describe both fixed-template and fixed-dimensional mechanisms. We now consider a different many-template regime, in which both the ambient dimension $d$ and the number of initialized templates $K$ grow. This regime is closely related to the high-dimensional result of~\cite[Theorem~3.11]{balanov2025confirmation}. When the templates are nearly orthogonal, the Gaussian hard-assignment estimator aligns with the template that generated its cell. By the Gaussian-spherical direction equivalence established above, the same conclusion holds for the one-step population update of spherical $K$-means.

\begin{proposition}[High-dimensional many-template alignment under cyclic symmetry]
\label{prop:hd_many_template_alignment}
Let $K=K(d)\to\infty$ as $d\to\infty$, and let
$\mu_0,\ldots,\mu_{K-1}\in \mathbb{S}^{d-1}$ be initialized templates satisfying
\begin{align}
    \max_{\ell_1\neq \ell_2 \in [K]} \langle \mu_{\ell_1},\mu_{\ell_2}\rangle \log(|\ell_1-\ell_2|) \longrightarrow 0,
    \label{eq:hd_near_orthogonality}
\end{align}
as $d\to\infty$.
Assume in addition that the templates satisfy a cyclic relation, namely, $\langle \mu_{\ell_1}, \mu_{\ell_2} \rangle = \rho_{|\ell_1 - \ell_2|\mathsf{mod} K}  $ for every $0 \leq \ell_1, \ell_2 \leq K-1$.  Then, for every fixed $\ell$, $\langle \mu_\ell^\star,\mu_\ell\rangle \to 1$, or, equivalently, $d_S(\mu_\ell^\star,\mu_\ell)\to 0$.
\end{proposition}

\begin{proof}
Let $g\sim\mathcal{N}(0,I_d)$, let $R\coloneqq\arg\max_{j\in[K]}\langle g,\mu_j\rangle$, and recall the definition of the corresponding Gaussian population update $m_\ell^{\mathrm{G}} \coloneqq \mathbb E[g\vert R=\ell]$ from~\eqref{eqn:Gaussian_hardAssign_mean}. By cyclic symmetry, $\mathbb P(R=\ell)=1/K$ for every $\ell$. Moreover, under condition~\eqref{eq:hd_near_orthogonality}, the pointwise part of \cite[Theorem~3.11]{balanov2025confirmation} yields 
\begin{align}
    \langle m_\ell^{\mathrm{G}},\mu_\ell\rangle / (\sqrt{2\log K)} \to 1, \label{eqn:ref_balanov2025confirmation}
\end{align}
 for every fixed $\ell$. In particular, condition~\eqref{eq:hd_near_orthogonality} is a uniform high-dimensional analogue of Berman's logarithmic decorrelation condition for the extremes of stationary Gaussian sequences~\cite{berman1964limit,leadbetter2012extremes}.

It remains to control the norm of $m_\ell^{\mathrm{G}}$. Clearly, by~\eqref{eqn:ref_balanov2025confirmation},
\begin{align}
    \|m_\ell^{\mathrm{G}}\|_2 \geq \langle m_\ell^{\mathrm{G}},\mu_\ell\rangle = (1+o(1))\sqrt{2\log K}.
    \label{eqn:lower_bound}
\end{align}
For the reverse bound, let $a_\ell \coloneqq m_\ell^{\mathrm{G}} / \|m_\ell^{\mathrm{G}}\|_2$. Then $Z_\ell\coloneqq\langle g,a_\ell\rangle\sim\mathcal{N}(0,1)$, and 
\begin{align}
    \mathbb{E}\!\left[Z_\ell \vert R = \ell \right]
    &=    \mathbb{E}_{g\sim\mathcal{N}(0,I_d)}\!\left[ \left\langle g,a_\ell\right\rangle \,\middle|\, R = \ell \right] \notag\\
    &=
    \left\langle \mathbb{E}_{g\sim\mathcal{N}(0,I_d)}\!\left[g\vert R = \ell \right], a_\ell \right\rangle \notag\\
    &=
    \left\langle m_\ell^{G}, \frac{m_\ell^{G}}{\lVert m_\ell^{G}\rVert_2} \right\rangle = \lVert m_\ell^{G}\rVert_2.
\end{align}
Therefore, for every $\lambda>0$, Jensen's inequality gives
\begin{align}
    \exp\left( \lambda\|m_\ell^{G}\|_2 \right)
    &\leq
    \mathbb{E}\left[ e^{\lambda Z_\ell} \,\middle|\, R=\ell \right] \notag\\
    &\leq
    K\mathbb{E}\left[e^{\lambda Z_\ell}\right] = K e^{\lambda^2/2}.
\end{align}
Taking natural logarithms and dividing by $\lambda$ yields
\begin{align}
    \|m_\ell^{G}\|_2 \leq \frac{\log K}{\lambda}+\frac{\lambda}{2}.
\end{align}
Choosing $\lambda=\sqrt{2\log K}$ yields $\|m_\ell^{G}\|_2 \leq \sqrt{2\log K}$. Combining with~\eqref{eqn:ref_balanov2025confirmation}--\eqref{eqn:lower_bound} yields $\langle m_\ell^{\mathrm{G}},\mu_\ell\rangle / \|m_\ell^{\mathrm{G}}\|_2 \to 1$. By Lemma~\ref{cor:spherical_gaussian_same_direction}, $\mu_\ell^\star = m_\ell^{\mathrm{G}} / \|m_\ell^\mathrm{G}\|_2$, which proves $\langle\mu_\ell^\star,\mu_\ell\rangle\to1$.
\end{proof}

\section{Fixed dimension and growing number of templates}
\label{sec:fixed_dimension}

The previous section showed that several results from the Gaussian confirmation-bias model transfer to spherical $K$-means through the Gaussian-spherical direction equivalence. In particular, Proposition~\ref{prop:hd_many_template_alignment} gives the spherical analogue of the high-dimensional Gaussian result of~\cite{balanov2025confirmation}: when both $d$ and $K$ grow and the initialized templates are sufficiently decorrelated, the one-step population update aligns with the corresponding initialized template.

Here, we focus on a complementary regime in which we fix the dimension $d$ and let the number of templates $K$ grow. This regime is particularly relevant to cryo-EM refinement, where the dimensionality of the image representation is fixed by the chosen discretization, while a large bank of reference templates can be generated from one or more candidate volumes to densely sample viewing directions and structural hypotheses~\cite{cheng2015primer, bendory2020single}. In fixed dimension, near-orthogonality cannot persist as $K\to\infty$. The relevant mechanism is instead geometric: the initialization induces a spherical Voronoi tessellation, and as $K$ increases, the cells become smaller. Consequently, the normalized centroid of each cell becomes increasingly close to the template that generated it. This gives the fixed-dimensional alignment rates proved below.

Proposition~\ref{prop:hardAssignmentPopulationLimit} implies that one population spherical $K$-means update is the normalized centroid of the Voronoi cell induced by the initialization. Thus, to understand how strongly the update preserves the initialization, it suffices to control the size of these cells. For random templates, this leads to two complementary bounds. The first is a typical-cell bound, controlling the average alignment error over the $K$ templates. The second is a worst-cell bound, controlling the largest alignment error over all templates. The typical-cell rate follows from standard estimates for random spherical Voronoi tessellations; see, for example, the analysis of typical spherical Voronoi cells in~\cite{kabluchko2021typical}.

\begin{proposition}[Typical-cell alignment]
\label{prop:typical_cell_rate}
Fix $d\geq2$, and let
$\mu_0,\ldots,\mu_{K-1}\stackrel{\mathrm{i.i.d.}}{\sim}\mathrm{Unif}(\mathbb{S}^{d-1})$.
Let $\mu_\ell^\star$ be the population update in~\eqref{eq:population-voronoi-centroid}.
Then, 
\begin{align}
    \mathbb{E}\left[ \frac{1}{K} \sum_{\ell=0}^{K-1}  d_S^2(\mu_\ell^\star,\mu_\ell) \right] = O\left(K^{-\frac{2}{d-1}}\right).    \label{eq:typical_cell_average_rate}
\end{align}
\end{proposition}

\begin{proof}
For each $\ell$, let $R_\ell \triangleq \sup_{v\in\mathcal{V}_\ell} d_S(v,\mu_\ell)$ denote the geodesic radius of the Voronoi cell generated by $\mu_\ell$.
By Lemma~\ref{lem:population_update_in_voronoi_cell},
$\mu_\ell^\star\in\mathcal{V}_\ell$.
Therefore, $d_S(\mu_\ell^\star,\mu_\ell) \leq R_\ell$, and hence $d_S^2(\mu_\ell^\star,\mu_\ell) \leq R_\ell^2$. 

By exchangeability of the random templates,
\begin{align}
    \mathbb{E}\left[ \frac{1}{K} \sum_{\ell=0}^{K-1} d_S^2(\mu_\ell^\star,\mu_\ell)\right] = \mathbb{E}\left[ d_S^2(\mu_0^\star,\mu_0) \right]  \leq \mathbb{E}[R_0^2].
\end{align}
Proposition~\ref{prop:typical_cell_radius} gives $\mathbb{E}[R_0^2] = O\left(K^{-2/ (d-1)}\right)$. Combining the two bounds proves the result.
\end{proof}

\begin{figure*}[t!]
\centering
\begin{minipage}[t]{0.68\textwidth}
    \vspace{0pt}
    \includegraphics[width=\linewidth]{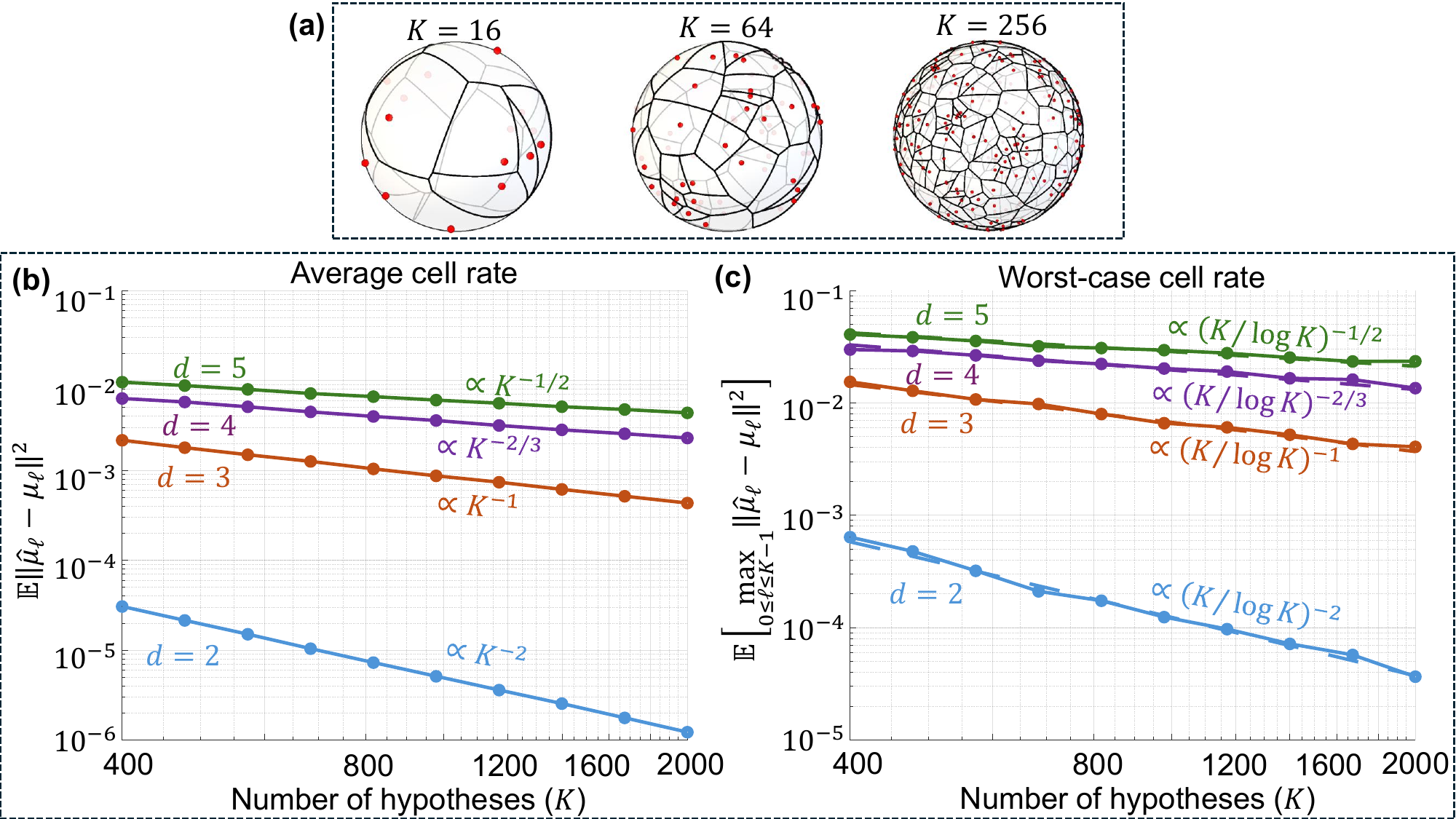}
\end{minipage}\hfill
\begin{minipage}[t]{0.31\textwidth}
    \vspace{0pt}
    \captionsetup{width=\linewidth}
    \caption{\textbf{Initialization-induced bias in the fixed-dimensional many-template regime.}
    (a) Random templates on $\mathbb{S}^2$ induce a spherical Voronoi tessellation that becomes finer as the number of postulated components $K$ increases. Red points denote the initialized templates, and black curves denote the corresponding Voronoi-cell boundaries (the gray cells). 
    (b) Average population alignment error as a function of $K$. The observed decay agrees with the typical-cell rate $K^{-2/(d-1)}$.
    (c) Worst-cell population alignment error.
    The observed decay agrees with the covering-rate scaling $(\log K/K)^{2/(d-1)}$.}    \label{fig:2}
\end{minipage}
\end{figure*}

The typical-cell result describes the average alignment error over templates. We next control all templates simultaneously. For this purpose, we use the random covering radius of the initialized templates:
\begin{align}
    \rho_K \triangleq \sup_{v\in \mathbb{S}^{d-1}} \min_{0\leq \ell\leq K-1} d_S(v,\mu_\ell). \label{eq:covering_radius_def}
\end{align}
The quantity $\rho_K$ is the largest geodesic distance from any point on the sphere to its nearest initialized template. Thus, it controls the largest radius of any spherical Voronoi cell induced by the templates.

For i.i.d.\ uniform templates on $\mathbb{S}^{d-1}$, Reznikov and Saff~\cite{reznikov2016covering} show that, as $K\to\infty$,
\begin{align}
    \rho_K = O_{\mathbb{P}} \left( \left(\frac{\log K}{K}\right)^{\frac{1}{d-1}} \right).    \label{eqn:covering_radius_bound}
\end{align}
Thus, the following result is an asymptotic statement in the large-$K$ limit, for fixed dimension $d$. It gives the corresponding worst-cell alignment bound.

\begin{thm}[Worst-cell alignment]
\label{thm:hardAssignmentAsymptoticLandFixedD}
Fix $d\geq2$, and let $\mu_0,\ldots,\mu_{K-1}\stackrel{\mathrm{i.i.d.}}{\sim}\mathrm{Unif}(\mathbb{S}^{d-1})$. Let $\mu_\ell^\star$ be the population update in~\eqref{eq:population-voronoi-centroid}. Then, as $K\to\infty$,
\begin{align}
    \max_{0\leq\ell\leq K-1} d_S^2(\mu_\ell^\star,\mu_\ell) = O_{\mathbb{P}}\left( \left(\frac{\log K}{K}\right)^{\frac{2}{d-1}} \right).    \label{eq:worst_cell_squared_alignment_rate}
\end{align}
\end{thm}

\begin{proof}
By Lemma~\ref{lem:population_update_in_voronoi_cell}, $\mu_\ell^\star\in\mathcal{V}_\ell$ for every $\ell$. Therefore, since $\mu_\ell$ is a nearest template to every point in $\mathcal{V}_\ell$,
\begin{align}
    d_S(\mu_\ell^\star,\mu_\ell) = \min_{0\leq k\leq K-1} d_S(\mu_\ell^\star,\mu_k)\leq \rho_K.
\end{align}
Hence, $\max_{0\leq\ell\leq K-1} d_S^2(\mu_\ell^\star,\mu_\ell) \leq \rho_K^2$. Combining this with the random covering-radius~\eqref{eqn:covering_radius_bound} proves~\eqref{eq:worst_cell_squared_alignment_rate}.
\end{proof}

The two rates capture different aspects of the same geometric mechanism. The expected average squared alignment error decays at the cell-volume scale $K^{-2/(d-1)}$. For the worst template, an additional logarithmic factor controls the largest cell in the random tessellation. In both cases, increasing the number of postulated components makes the population update more aligned with the initialization, even though the observations contain no cluster structure.

Figure~\ref{fig:2} illustrates the geometric refinement of the initialization-induced Voronoi tessellation and the corresponding population alignment rates. Panel~(a) shows random templates on $\mathbb{S}^2$ and the spherical Voronoi cells they induce for increasing values of $K$. Panels~(b)-(c) show the average-cell and worst-cell population errors, respectively, for dimensions $d=2,3,4,5$. The simulations use $K$ logarithmically spaced between $400$ and $2000$ with $10$ values, and $30$ independent trials for each pair $(d,K)$. For $d=2$, the population update is computed exactly from the circle geometry. For $d\geq3$, the population Voronoi-centroid integrals are approximated using $2\times10^5$ random integration points on $\mathbb{S}^{d-1}$.

\section{Consequences and Extensions}
\label{sec:extensions}

This section presents three consequences of our main results and concludes with directions for future work. 

\subsection{Multiple iterations}
\label{subsec:multiple_iterations}

The results so far have focused on a single spherical $K$-means iteration, which is sufficient to expose the initialization-induced mechanism under the uniform model. The extension to multiple iterations is natural and parallels the Gaussian analysis in~\cite[Section~IV]{balanov2025confirmation}. At each iteration, the current templates induce a spherical Voronoi tessellation, and the next templates are the normalized centroids of the corresponding cells. Hence, if each update remains close to the current template, then the iterates remain close to the original initialization for any fixed number of iterations.

Indeed, if $\langle \mu_\ell^{(t+1)},\mu_\ell^{(t)}\rangle \geq 1-\varepsilon_K$, for $t=0,\ldots,T-1$, and for small enough $\varepsilon_K$, then the argument of~\cite[Proposition~4.1]{balanov2025confirmation} gives
\begin{align}
    \langle \mu_\ell^{(T)},\mu_\ell^{(0)}\rangle \geq 1-T^2\varepsilon_K .
\end{align}
Under the i.i.d.\ uniform initialization considered in Section~\ref{sec:fixed_dimension}, our one-step results provide precisely such a vanishing small-movement bound for the first update: as $K\to\infty$, the initialization-induced Voronoi cells become small and the population update becomes increasingly aligned with its generating template. Thus, for fixed $T$, if an analogous bound with $\varepsilon_K\to 0$ continues to hold along the subsequent iterates, the $T$-th iterate remains asymptotically aligned with the initialization.

\subsection{Small nonzero concentration}
\label{subsec:small_kappa}

Our main results concern the uniform case $\kappa=0$, in which the observations are uniformly distributed on the sphere and do not contain any directional signal. This setting isolates initialization-induced bias in its purest form. The same mechanism, however, persists for weak directional signal, because a vMF mixture with small concentration is a perturbation of the uniform distribution. This weak-signal regime more closely reflects realistic cryo-EM settings, where the observations may contain genuine structural information, but each individual image remains highly noisy and only weakly informative~\cite{bendory2020single}. Accordingly, the one-step population update for small $\kappa>0$ remains close to the null Voronoi-centroid update, as also illustrated in Figure~\ref{fig:1}(c). 

To formalize this observation, consider the $K$-component vMF mixture on $\mathbb{S}^{d-1}$ defined in~\eqref{eq:vMF-mixture-density}. For notational simplicity, we write its density with respect to the normalized surface measure $d\sigma(v)$ as
\begin{align}
    p_\kappa(v) = \sum_{j=0}^{K-1} w_j f_{\mathrm{vMF}}(v;\theta_j,\kappa_j),
\end{align}
where the directions $\theta_j$ and weights $w_j$ are arbitrary and $\max_{0\leq j\leq K-1}\kappa_j \leq \kappa$.

Fix initialized templates $\mu_0,\ldots,\mu_{K-1}$, and let $\mathcal{V}_\ell$ be the Voronoi cell defined in~\eqref{eqn:VlDef}. Define the unnormalized population update under $p_\kappa$ by
\begin{align}
    m_\ell(\kappa) \triangleq \int_{\mathcal{V}_\ell} v p_\kappa(v)\,d\sigma(v),
\end{align}
and the corresponding spherical $K$-means population update by
\begin{align}
    \mu_\ell^\star(\kappa) \triangleq \frac{m_\ell(\kappa)} {\|m_\ell(\kappa)\|_2}.
\end{align}
At $\kappa=0$, we have $p_0(v)=1$, and therefore $m_\ell(0) = \int_{\mathcal{V}_\ell} v\,d\sigma(v)$, so $\mu_\ell^\star(0)$ coincides with the population update in~\eqref{eq:population-voronoi-centroid}.

\begin{proposition}[Stability of the one-step update for small $\kappa$]
\label{prop:small_kappa_stability}
Fix $d\geq2$ and initialized templates $\mu_0,\ldots,\mu_{K-1}$. Assume that $m_\ell(0)\neq 0$. Then, for sufficiently small $\kappa$, $\|m_\ell(\kappa)-m_\ell(0)\|_2 = O (\kappa)$. Consequently,
\begin{align}
    \|\mu_\ell^\star(\kappa)-\mu_\ell^\star(0)\|_2 = O(\kappa).
\end{align}
\end{proposition}

\begin{proof}
For small $\kappa$, the vMF density~\eqref{eqn:vmf_density} admits the uniform expansion $f_{\mathrm{vMF}}(v;\theta,\kappa) = 1+O(\kappa)$, uniformly over $v,\theta\in \mathbb{S}^{d-1}$. Since $p_\kappa$ is a mixture satisfying $\max_j\kappa_j\leq\kappa$, it follows that $p_\kappa(v) = 1+O(\kappa)$ uniformly over $v\in \mathbb{S}^{d-1}$. Therefore,
\begin{align}
    m_\ell(\kappa)-m_\ell(0) &= \int_{\mathcal{V}_\ell} v\bigl(p_\kappa(v)-1\bigr)\,d\sigma(v),
\end{align}
and hence
\begin{align}
    \|m_\ell(\kappa)-m_\ell(0)\|_2  &\leq \int_{\mathcal{V}_\ell}  \bigl|p_\kappa(v)-1\bigr|\,d\sigma(v) \\
    &=  O(\kappa\,\sigma(\mathcal{V}_\ell)).
\end{align}

The bound for the normalized updates follows from
\begin{align}
    \left\| \frac{x}{\|x\|_2} - \frac{y}{\|y\|_2} \right\|_2 \leq \frac{2\|x-y\|_2}{\|y\|_2},
\end{align}
valid for $y\neq0$ and $x$ sufficiently close to $y$. Applying this inequality with $x=m_\ell(\kappa)$ and $y=m_\ell(0)$ gives
\begin{align}
    \|\mu_\ell^\star(\kappa)-\mu_\ell^\star(0)\|_2 = O\left( \frac{\kappa\,\sigma(\mathcal{V}_\ell)} {\|m_\ell(0)\|_2} \right),
\end{align}
which proves the result.
\end{proof}

\subsection{Explicit population dynamics on the circle}
\label{sec:circle_population_dynamics}

The case $d=2$ gives a closed-form description of the population dynamics under the uniform model. We identify $\mathbb{S}^1$ with complex unit vectors and write $\mu_\ell(t)=e^{\mathrm{i}\theta_\ell(t)}$, for $\ell=0,\ldots,K-1$. At each iteration, the angles are ordered cyclically as
\begin{align}
    0\leq \theta_0(t)<\theta_1(t)<\cdots<\theta_{K-1}(t)<2\pi .
\end{align}

Under the uniform model, the Voronoi cell of $\mu_\ell(t)$ is the arc whose endpoints are the midpoints between neighboring templates, i.e., if
\begin{align}
    g_\ell(t) = \theta_{\ell+1}(t)-\theta_\ell(t)
    \label{eqn:cyclic_gaps}
\end{align}
denotes the cyclic gap, then the cell of $\mu_\ell(t)$ is the arc
\begin{align}
    \left[ \theta_\ell(t)-\frac{g_{\ell-1}(t)}{2}, \theta_\ell(t)+\frac{g_\ell(t)}{2} \right].
\end{align}
The normalized centroid of an arc points in the direction of its angular midpoint. Therefore, one population spherical $K$-means update satisfies
\begin{align}
    \theta_\ell(t+1)
    &=
    \theta_\ell(t) + \frac{g_\ell(t)-g_{\ell-1}(t)}{4} \nonumber\\
    &= \frac{ \theta_{\ell-1}(t) + 2\theta_\ell(t) + \theta_{\ell+1}(t)}{4}.  \label{eq:s1_theta_update}
\end{align}
Thus, on the circle, the population spherical $K$-means map acts as a local averaging operator on neighboring template angles. Combining~\eqref{eq:s1_theta_update} with the definition of the cyclic gaps in~\eqref{eqn:cyclic_gaps} gives the following discrete gap dynamics.

\begin{proposition}[Population gap dynamics on $\mathbb{S}^1$]
\label{prop:s1_gap_dynamics}
Under the uniform model on $\mathbb{S}^1$, the cyclic gaps $g_\ell(t)$~\eqref{eqn:cyclic_gaps} evolve according to
\begin{align}
    g_\ell(t+1) = \frac{ g_{\ell-1}(t) + 2g_\ell(t) + g_{\ell+1}(t) }{4}.
    \label{eq:s1_gap_update}
\end{align}
\end{proposition}

Consequently, the gap vector follows a discrete heat equation on the cycle. Let $\bar{g}=2\pi/K$ be the uniform gap size. The update in~\eqref{eq:s1_gap_update} is a circulant averaging operator. Its Fourier eigenvalues are
\begin{align}
    \lambda_q = \frac{1}{2} + \frac{1}{2} \cos\left(\frac{2\pi q}{K}\right) = \cos^2\left(\frac{\pi q}{K}\right),
\end{align}
for $q=0,\ldots,K-1$. The eigenvalue $\lambda_0=1$ preserves the total length $\sum_{\ell=0}^{K-1}g_\ell(t)=2\pi$. All other modes contract. In particular,
\begin{align}
    \|g(t)-\bar{g}\mathbf{1}\|_2 \leq \cos^{2t}\left(\frac{\pi}{K}\right) \|g(0)-\bar{g}\mathbf{1}\|_2 . \label{eq:s1_gap_contraction}
\end{align}
Since $\cos^2\left(\frac{\pi}{K}\right) = 1-\Theta(K^{-2})$ as $K \to \infty$, the relaxation time of the population dynamics is of order $K^2$.

This explicit dynamics has two consequences. First, as $t\to\infty$, the gaps converge to the uniform gap vector $\bar{g}\mathbf{1}$,  that is, the templates converge to a regular $K$-gon, up to a global rotation. Indeed, any fixed point of~\eqref{eq:s1_gap_update} must satisfy $g_{\ell-1}-2g_\ell+g_{\ell+1}=0$ on the cycle, whose only cyclic solutions are constant. Thus, the only nondegenerate population fixed points on $\mathbb{S}^1$ are the uniformly spaced configurations.

Second, convergence to this regular configuration can be very slow when $K$ is large. For iteration numbers $t\ll K^2$, the contraction factor in~\eqref{eq:s1_gap_contraction} is close to one, so the gap pattern remains close to its initialized configuration. Thus, under the uniform model, spherical $K$-means does not immediately erase the initialization. Rather, it gradually regularizes the initialization-induced Voronoi tessellation through a heat-flow-like dynamics. This provides an explicit multi-iteration explanation for the persistence of initialization-induced structure in the many-template regime.

\subsection{Future work}

Several directions remain open. A natural next step is to study soft-assignment methods, particularly expectation-maximization (EM) for vMF mixtures, and determine whether they weaken or amplify initialization bias relative to spherical $K$-means. The corresponding Gaussian problem was analyzed in~\cite{balanov2025confirmation}, providing a natural starting point for the directional setting. A second direction is to develop finite-sample theory quantifying when the empirical iterates track their population counterparts as $n$, $K$, and $d$ vary. Finally, extending the explicit $\mathbb{S}^1$ population dynamics to higher-dimensional spheres could clarify the long-term behavior of repeated updates, including their fixed points and convergence rates, as well as the connection between heat-flow-like dynamics and weak-signal initialization bias.

\bibliographystyle{IEEEbib}

\begin{appendix}
\section{Radius of a  spherical Voronoi cell}

The next proposition quantifies the size of a typical spherical Voronoi cell generated by $K$ independent uniform templates. It shows that the cell radius has an exponentially decaying tail at the natural scale $K^{-1/(d-1)}$, and consequently that its mean squared radius is of order $K^{-2/(d-1)}$.

\begin{proposition}[Radius of a typical spherical Voronoi cell]
\label{prop:typical_cell_radius}
Fix $d\geq 2$, and let $\mu_0,\ldots,\mu_{K-1}\stackrel{\mathrm{i.i.d.}}{\sim} \mathrm{Unif}(\mathbb{S}^{d-1})$. Let $\mathcal{V}_0$ be the spherical Voronoi cell generated by $\mu_0$, and define $R_0 \triangleq \sup_{v\in\mathcal{V}_0} d_S(v,\mu_0)$. Then there exist constants $c_d,C_d,r_d>0$, depending only on $d$, such that, for every $K\geq2$ and $0<r\leq r_d$,
\begin{align}
    \mathbb{P}(R_0>r) \leq C_d\exp\!\left(-c_dKr^{d-1}\right).    \label{eq:typical_cell_radius_tail}
\end{align}
Consequently, as $K\to\infty$ with $d$ fixed,
\begin{align}
    \mathbb{E}[R_0^2] = O\!\left(K^{-\frac{2}{d-1}}\right).    \label{eq:typical_cell_radius_second_moment}
\end{align}
\end{proposition}

\subsection{Notation and definitions}
Before proving Proposition~\ref{prop:typical_cell_radius}, we introduce the geometric construction used in the argument. For $\mu\in\mathbb{S}^{d-1}$, let
\begin{align}
    \mathbb{S}_{\mu}^{d-2} \triangleq \left\{\eta\in\mathbb{S}^{d-1}: \langle\eta,\mu\rangle=0\right\}
\end{align}
denote the unit tangent sphere at $\mu$. For $\eta\in\mathbb{S}_{\mu}^{d-2}$, the Riemannian exponential map is
\begin{align}
    \operatorname{Exp}_{\mu}(t\eta) = \cos(t)\mu+\sin(t)\eta,
    \qquad  0\leq t\leq\pi.  \label{eq:spherical_exponential_map}
\end{align}
Thus, a point at geodesic distance $t$ from $\mu$ can be represented by a tangent direction $\eta$ and a radial distance $t$; the tangent direction is unique for $0<t<\pi$ and nonunique at the antipode $t=\pi$. 

We adapt the finite-direction argument used in the proof of~\cite[Lemma~5.1]{penrose2007laws}. Condition on $\mu_0$ and fix $\varepsilon=\pi/12$. Choose a finite $\varepsilon$-net $\{\xi_j\}_{j=1}^{N_d}\subset\mathbb{S}_{\mu_0}^{d-2}$, where $N_d$ depends only on $d$. Thus, the caps
\begin{align}
    A_j \triangleq \left\{\eta\in\mathbb{S}_{\mu_0}^{d-2}: d_S(\eta,\xi_j)\leq\varepsilon \right\},
\end{align}
for $ j=1,\ldots,N_d$, cover $\mathbb{S}_{\mu_0}^{d-2}$. 

For $r>0$, define the corresponding truncated geodesic sectors
\begin{align}
    \mathcal{C}_j(r) \triangleq \left\{\operatorname{Exp}_{\mu_0}(s\eta): 0<s\leq r,\; \eta\in A_j\right\} \subseteq \mathbb{S}^{d-1}.    \label{eq:def_truncated_spherical_sector}
\end{align}
Each $\mathcal{C}_j(r)$ is contained in the closed geodesic ball
\begin{align}
    \mathcal{B}_S(\mu_0,r) \triangleq \left\{ v\in\mathbb{S}^{d-1}: d_S(v,\mu_0)\leq r \right\}
\end{align}
Moreover, since the caps $A_j$ cover all tangent directions at $\mu_0$, for $0<r<\pi$ the corresponding sectors collectively cover the entire punctured geodesic ball:
\begin{align}
    \bigcup_{j=1}^{N_d}\mathcal{C}_j(r) = \mathcal{B}_S(\mu_0,r)\setminus\{\mu_0\}.    \label{eq:sectors_cover_geodesic_ball}
\end{align}
Thus, the finite collection of sectors provides a directional decomposition of the geodesic neighborhood of $\mu_0$.

Finally, let
\begin{align}
    E_j(r) \triangleq \left\{ \mathcal{C}_j(r) \cap \{\mu_1,\ldots,\mu_{K-1}\} = \varnothing \right\}
\end{align}
be the event that the $j$th sector contains no competing template within distance $r$.

\subsection{Auxiliary results}
The next lemma quantifies the size of each truncated directional sector and, in turn, the probability that it contains no competing template. Its key point is that a sector of geodesic radius $r$ has surface measure of order $r^{d-1}$, so the probability that all $K-1$ competing templates miss it decays exponentially in $Kr^{d-1}$.

\begin{lem}
\label{lem:spherical_sector_mass}
Let $\sigma$ denote normalized surface measure on $\mathbb{S}^{d-1}$. There exist constants $a_d,r_d>0$, depending only on $d$, such that for every $0<r\leq r_d$,
\begin{align}
    \sigma\!\left(\mathcal{C}_j(r)\right) \geq a_d r^{d-1},    \label{eq:spherical_sector_measure}
\end{align}
uniformly in $\mu_0$ and $j$. Consequently, for $ 0<r\leq r_d$,
\begin{align}
    \mathbb{P}\!\left(E_j(r)\vert\mu_0\right) \leq  \exp\!\left(-a_d(K-1)r^{d-1}\right).    \label{eq:empty_sector_tail}
\end{align}
\end{lem}

\begin{proof}
Since $A_j$ has fixed angular radius $\varepsilon$, its measure on the unit tangent sphere is a positive constant depending only on $d$. In geodesic polar coordinates centered at $\mu_0$, the radial surface element is proportional to $\sin^{d-2}(s)\,ds$. Since $\sin(s)\asymp s$ uniformly for sufficiently small $s$, integration over $A_j$ and $0<s\leq r$ gives \eqref{eq:spherical_sector_measure}.

Conditionally on $\mu_0$, the templates $\mu_1,\ldots,\mu_{K-1}$ remain independent and uniformly distributed on $\mathbb{S}^{d-1}$, while $\mathcal{C}_j(r)$ is fixed. Since $E_j(r)$ is the event that none of these $K-1$ templates falls inside $\mathcal{C}_j(r)$, for every $k=1,\ldots,K-1$,
\begin{align}
    \mathbb{P}\!\left( \mu_k\notin\mathcal{C}_j(r) \vert \mu_0 \right) = 1-\sigma\!\left(\mathcal{C}_j(r)\right).
\end{align}
Therefore, by conditional independence,
\begin{align}
    \mathbb{P}\!\left(E_j(r)\mid\mu_0\right)
    &=
    \left(1-\sigma\!\left(\mathcal{C}_j(r)\right)\right)^{K-1} \\
    &\leq
    \exp\!\left(-(K-1)\sigma\!\left(\mathcal{C}_j(r)\right) \right),
\end{align}
where we used $1-x\leq e^{-x}$. Equation~\eqref{eq:spherical_sector_measure} then yields \eqref{eq:empty_sector_tail}.
\end{proof}

The next lemma formalizes the geometric observation that if the Voronoi cell of $\mu_0$ reaches beyond geodesic distance $r$, then at least one of the finitely many directional sectors must contain no competing template within distance $r$ of $\mu_0$.

\begin{lem}[Directional blocking and the empty-sector criterion]
\label{lem:large_cell_empty_sector}
For every $0 < r < \pi$,
\begin{align}
    \{R_0>r\} \subseteq  \bigcup_{j=1}^{N_d} E_j(r). \label{eq:radius_event_inclusion}
\end{align}
\end{lem}

\begin{proof}
Suppose that none of the events $E_j(r)$ occurs. Assume, toward a contradiction, that $R_0>r$. Then there exists $v\in\mathcal{V}_0$ with $t\triangleq d_S(v,\mu_0)>r$. Write
\begin{align}
    v = \operatorname{Exp}_{\mu_0}(t\eta_v)
    \label{eqn:def_v}
\end{align}
for some $\eta_v\in\mathbb{S}_{\mu_0}^{d-2}$. Since the caps $A_j$ cover the unit tangent sphere, there exists $j$ such that $\eta_v\in A_j$.

Because $E_j(r)$ does not occur, the sector $\mathcal{C}_j(r)$ contains a competing template, say
\begin{align}
    \mu_k = \operatorname{Exp}_{\mu_0}(s\eta_k) \in \mathcal{C}_j(r),
    \label{eqn:def_mu_k}
\end{align}
where $0<s\leq r<t$ and $\eta_k\in A_j$. Since both $\eta_v$ and $\eta_k$ belong to $A_j$,
\begin{align}
    d_S(\eta_v,\eta_k) &\leq d_S(\eta_v,\xi_j) + d_S(\xi_j,\eta_k) \\
    &\leq 2\varepsilon = \frac{\pi}{6}.    \label{eq:sector_direction_separation}
\end{align}

We now show that the closer template $\mu_k$, lying in approximately the same direction from $\mu_0$ as $v$, is closer to $v$ than $\mu_0$ is. Let $\theta\triangleq d_S(\eta_v,\eta_k)\leq\pi/6$. Recall from~\eqref{eqn:def_v} and~\eqref{eqn:def_mu_k} that $v$ and $\mu_k$ are obtained from the same base point $\mu_0$ by moving geodesic distances $t$ and $s$, respectively, along the tangent directions $\eta_v$ and $\eta_k$:
\begin{align}
    v &= \cos(t)\mu_0+\sin(t)\eta_v, \\
    \mu_k &= \cos(s)\mu_0+\sin(s)\eta_k.
\end{align}
Since $\eta_v,\eta_k\in\mathbb{S}_{\mu_0}^{d-2}$, both tangent directions are orthogonal to $\mu_0$, so $\langle\mu_0,\eta_v\rangle = \langle\mu_0,\eta_k\rangle = 0$. Moreover, by the definition of $\theta$, $\langle\eta_v,\eta_k\rangle = \cos(\theta)$. Using $\cos d_S(x,y)=\langle x,y\rangle$ for $x,y\in\mathbb{S}^{d-1}$, we therefore obtain
\begin{align}
    \cos d_S(v,\mu_k) &= \langle v,\mu_k\rangle \\ \notag
    &=
    \left\langle \cos(t)\mu_0+\sin(t)\eta_v,\, \cos(s)\mu_0+\sin(s)\eta_k \right\rangle \\ \notag
    &= \cos(t)\cos(s) + \sin(t)\sin(s)\cos(\theta).    \label{eq:blocking_cosine_vk}
\end{align}
Similarly,
\begin{align}
    \cos d_S(v,\mu_0) = \langle v,\mu_0\rangle = \cos(t),
\end{align}
which is consistent with $d_S(v,\mu_0)=t$ in~\eqref{eqn:def_v}. Subtracting the two identities gives
\begin{align}
    &\cos d_S(v,\mu_k)-\cos d_S(v,\mu_0)
    \notag\\
    &\qquad = \sin(t)\sin(s)\cos(\theta) - \cos(t)\bigl(1-\cos(s)\bigr).    \label{eq:blocking_cosine_difference}
\end{align}

If $t\geq\pi/2$, then $\cos(t)\leq0$, while $\sin(t)\geq0$, $\sin(s)>0$, $\cos(\theta)>0$, and $1-\cos(s)>0$. Hence the right-hand side of \eqref{eq:blocking_cosine_difference} is strictly positive. It remains to consider $t<\pi/2$. Using $0<s<t$ and the monotonicity of $\tan(\cdot)$ on $(0,\pi/2)$,
\begin{align}
    \tan\!\left(\frac{s}{2}\right) < \tan\!\left(\frac{t}{2}\right) <\frac{1}{2}\tan(t) <  \cos(\theta)\tan(t),
    \label{eq:blocking_tangent_inequality}
\end{align}
where the second inequality follows from
\begin{align}
    \tan\!\left(\frac{t}{2}\right) = \frac{\sin(t)}{1+\cos(t)} < \frac{\sin(t)}{2\cos(t)} = \frac{1}{2}\tan(t),
\end{align}
and the last inequality follows from $\cos(\theta)\geq\cos(\pi/6)=\sqrt{3}/2>1/2$. Using
\begin{align}
    \tan\!\left(\frac{s}{2}\right) = \frac{1-\cos(s)}{\sin(s)},
\end{align}
and noting that $\cos(t)>0$ and $\sin(s)>0$, \eqref{eq:blocking_tangent_inequality} is equivalent to
\begin{align}
    \sin(t)\sin(s)\cos(\theta) - \cos(t)\bigl(1-\cos(s)\bigr) > 0.
\end{align}
Therefore, by~\eqref{eq:blocking_cosine_difference}, $\cos d_S(v,\mu_k) > \cos d_S(v,\mu_0)$. Since cosine is strictly decreasing on $[0,\pi]$, it follows that
\begin{align}
    d_S(v,\mu_k) < d_S(v,\mu_0).
\end{align}
This contradicts $v\in\mathcal{V}_0$, since every point in $\mathcal{V}_0$ is at least as close to $\mu_0$ as to any competing template.

Therefore, if none of the events $E_j(r)$ occurs, then $R_0\leq r$. Taking the contrapositive yields
\begin{align}
    \{R_0>r\} \subseteq \bigcup_{j=1}^{N_d}E_j(r),
\end{align}
which proves~\eqref{eq:radius_event_inclusion}.
\end{proof}

\subsection{Proof of Proposition~\ref{prop:typical_cell_radius}}

By Lemma~\ref{lem:large_cell_empty_sector} and the union bound, conditionally on $\mu_0$, for every $0<r\leq r_d$,
\begin{align}
    \mathbb{P}(R_0>r\vert \mu_0)
    &\leq
    \sum_{j=1}^{N_d} \mathbb{P}(E_j(r) \vert \mu_0) \\
    &\leq N_d \exp\!\left(-a_d(K-1)r^{d-1}\right),
\end{align}
where the second inequality follows from Lemma~\ref{lem:spherical_sector_mass}.
Since $K-1\geq K/2$ for $K\geq2$, and the bound is uniform in $\mu_0$, averaging over $\mu_0$ yields, after absorbing constants depending only on $d$,
\begin{align}
    \mathbb{P}(R_0>r) \leq C_d\exp\!\left(-c_dKr^{d-1}\right).
\end{align}
This proves~\eqref{eq:typical_cell_radius_tail}.

For the second-moment bound, since $0\leq R_0\leq\pi$, the tail integration formula gives
\begin{align}
    \mathbb{E}[R_0^2] = 2\int_0^\pi r\,\mathbb{P}(R_0>r)\,dr.
\end{align}
For $r\geq r_d$, monotonicity of the tail probability and \eqref{eq:typical_cell_radius_tail} give
\begin{align}
    \mathbb{P}(R_0>r) \leq \mathbb{P}(R_0>r_d) = O\!\left(e^{-c_d'K}\right)
\end{align}
for some $c_d'>0$.
Therefore,
\begin{align}
    \mathbb{E}[R_0^2] &\leq 2C_d\int_0^{r_d}r\exp\!\left(-c_dKr^{d-1}\right)\,dr + O\!\left(e^{-c_d'K}\right).  \label{eq:typical_cell_second_moment_integral}
\end{align}
With the change of variables
$u=c_dKr^{d-1}$,
\begin{align}
    \notag \int_0^{r_d}
    r\exp\!\left(-c_dKr^{d-1}\right)\,dr
    &= \frac{(c_dK)^{-\frac{2}{d-1}}}{d-1} \int_0^{c_dKr_d^{d-1}} u^{\frac{2}{d-1}-1}e^{-u}\,du \\
    &= O\!\left(K^{-\frac{2}{d-1}}\right).
\end{align}
Since the exponentially small remainder is negligible, we conclude that
\begin{align}
    \mathbb{E}[R_0^2] = O\!\left(K^{-\frac{2}{d-1}}\right),
\end{align}
which proves~\eqref{eq:typical_cell_radius_second_moment}.

\end{appendix}

\end{document}